%% file: main.tex
\documentclass[11pt]{article}
\usepackage{authblk}
\usepackage{amsmath}
\usepackage{amssymb,amsthm}
\usepackage{graphicx}
\usepackage{amstext}
\usepackage{wrapfig}
\usepackage{amsfonts}
\usepackage{epsfig}

\usepackage{hyperref}
\usepackage{graphicx}
\usepackage{subfigure}
\usepackage{microtype}

\graphicspath{{./figure/}}
\usepackage[utf8]{inputenc}
\usepackage{amsmath,amssymb,xspace,epsfig}
\usepackage{multicol, lipsum}
\usepackage{xcolor}
\usepackage{color}
\usepackage{multirow}

\newcommand{\reals}{\mathbb{R}}
\def\eps{\varepsilon}

\newcommand{\frechet}{Fr\'echet }
\newcommand{\etal}{ et al.}

\let\eps\varepsilon

\newcommand{\denselist}{\itemsep 0pt\parsep=0.8pt\partopsep 0pt}
\newcommand{\bitem}{\begin{itemize}\denselist}
\newcommand{\eitem}{\end{itemize}}
\newcommand{\benum}{\begin{enumerate}\denselist}
\newcommand{\eenum}{\end{enumerate}}

\newcommand{\overbar}[1]{\mkern 1.5mu\overline{\mkern-1.5mu#1\mkern-1.5mu}\mkern 1.5mu}

\newtheorem{theorem}{Theorem}	
\newtheorem{lemma}[theorem]{Lemma}

\newtheorem{definition}{Definition}

\newtheorem{remark}{Remark}
 
\usepackage[margin=1in]{geometry}

\title{Time Window \frechet and Metric-Based Edit Distance for Passively Collected Trajectories}

\author[1]{Jiaxin Ding\thanks{Email:jiaxinding@sjtu.edu.cn}}
\author[2]{Jie Gao}
\author[3]{Steven Skiena}
\affil[1]{Shanghai Jiao Tong University}
\affil[2]{Rutgers University}
\affil[3]{Stony Brook University}
\date{}                     
\date{}

\begin{document}
\maketitle

\begin{abstract}
The advances of modern localization techniques and the wide spread of mobile devices have provided us great opportunities to collect and mine human mobility trajectories. In this work, we focus on passively collected trajectories, which are sequences of time-stamped locations that mobile entities visit. To analyse such trajectories, a crucial part is a measure of similarity between two trajectories. 
We propose the time-window \frechet distance, which 
enforces the maximum temporal separation between  points of two trajectories that can be paired in the calculation of the \frechet distance, and the metric-based edit distance which incorporates the underlying metric in the computation of the insertion and deletion costs.  
Using these measures, we can cluster trajectories to infer group motion patterns. We look at the $k$-gather problem which requires each cluster to have at least $k$ trajectories. We prove that $k$-gather remains NP-hard under edit distance, metric-based edit distance and Jaccard distance. 
Finally, we improve over previous results on discrete \frechet distance and show that there is no strongly sub-quadratic time with approximation factor less than $1.61$ in two dimensional setting unless SETH fails. 
\end{abstract}

\newpage


\section{Introduction}

Advances in localization techniques and wireless technologies have allowed for the collection of a huge volume of trajectories of pedestrians and vehicles. 
The rich knowledge in human mobility data provides great opportunities to 
mine interesting patterns, that can be useful for numerous applications, including but not limited to traffic management, urban planning, and the scheduling of autonomous vehicles. 
In many scenarios, mobile users' trajectories are collected passively, i.e., inferred from connection traces with WiFi access points, cellular towers, or through transactions of credit cards or transit cards in public transportation systems.
These records consist of users' locations and the corresponding time stamps, so the sequence of records produces a good approximation of mobile user trajectory. 
Prior studies of trajectories collected passively include the study of human mobility traces of over $100,000$ mobile phone users~\cite{gonzalez08understanding} and the study of about $95,000$ users in a large university campus~\cite{mobicamp16}. 


Trajectories collected through passive sensing are unique in many ways. First the collected trajectories respect the density of wireless checkpoints (WiFi, cellular towers, etc) and have much lower resolution than GPS traces. Missing data is the norm, for example, in regions without WiFi access points. It is an interesting problem to properly handle the time-stamp information. 
First, time stamp labels cannot be ignored but it is preferred to allow flexibility in handling time stamps. Accurate spatial temporal information (i.e., the location of a person at a particular time) can be considered sensitive and private.
With long term trajectory data, frequent locations~\cite{gonzalez08understanding,song2010limits}, co-locations~\cite{Eagle08092009}, and specific patterns~\cite{citeulike:12207609,Ma:2010:PVP:1859995.1860017}  of users can be easily learned,  
which can be used to identify a user, breach their social ties and locate their whereabouts at any time. 
Thus location/time-stamp data in published data sets is often intentionally perturbed or generalized~\cite{Mokbel07}. Further, although human mobility shows great regularity and repetition, people have a fair amount of flexibility in daily routines. For example, it is common that a user goes to work every day but the time of leaving home may fluctuate.  

With passively collected trajectories, one can 
discover interesting mobility patterns by clustering trajectories into groups of similar ones, for which we need a measure of similarity.
There has been much work on  distances to measure similarity between curves. 
For example, Hausdorff distance measures the maximum distance of all points on one curve to their nearest point on the other curves. 
\frechet distance measures the minimum of maximum length between all possible pairing/coupling of points on the two curves.  

These distances for curves when applied to data mining in passively collected human trajectory data are not ideal.
First, all these distances focus on the \emph{shape} of the curves but ignore the time stamp information.  
But time dimension is important for understanding group mobility and traffic. Two individuals, traveling along the same route but at totally different time frames shall  not be considered similar. The duration spent at a location is a crucial attribute for important places for a user and the patten of visiting a sequence of locations is useful to infer the semantics of the trip.

In this paper we develop alternative similarity measures. We explored in two directions.
\begin{itemize}
\item \emph{Time-window \frechet distance.} In the first variant, we introduce time-window \frechet distance, where only pairs of points within a time window $\delta$ can be coupled. Further, we show an algorithm to compute the time window \frechet distance with running time related to the complexity of the input curves. 

\item \emph{Metric-based edit distance.} In the second variant, we take uniform sampling along the time dimension and use a combinatorial representation as a string of cells/towers/APs visited by the user and the edit distance -- the minimum number of changes (insertions or deletions) to turn one sequence to another. 
We consider a variant of edit distance by incorporating the underlying metric in the calculation of the edit distance. Specifically, the cost of inserting or deleting a symbol is based on the change of the metric distance before and after the operation. This will be able to handle the issue of missing data. We also discuss algorithms for computing the metric based edit distance. 
\end{itemize}

Next, we talk about clustering motion trajectories into meaningful clusters. We use the notion of $k$-gather clustering~\cite{aggarwal2006achieving}, which is to minimize the radius of any cluster, such that each trajectory is in a cluster of at least $k$ traces. 
The requirement of $k$ is needed to define group motion and often adopted to define popular/meaningful traffic patterns. There is also additional benefit of $k$-anonymity~\cite{Sweeney:2002:KAM:774544.774552} if only a summary of each cluster is reported/shared with the public -- one cannot tell a specific user from a group of $k$ trajectories. It is known that $k$-gathering for points in a metric space is NP-hard and there is a $2$-approximation algorithm~\cite{aggarwal2006achieving}. We extend the hardness proof to the case of trajectories under edit distance, metric based edit distance and the Jaccard distance.

Finally, we provide a hardness proof. It is recently shown that there is no subquadratic algorithm to compute \frechet distance unless SETH fails~\cite{bringmann2014walking} and later in \cite{bringmann2015approximability} it is shown that it is true for the approximation factor up to $1.399$ in one dimension. We show that in $\reals^2$ that there exists no subquadratic algorithm to approximate discrete \frechet distance with factor less than $1.61$ unless SETH fails. 


\input{related}

\input{frechetDefinition}
\input{timewindow}

\input{weight}

\section{\texorpdfstring{$k$}-gather Clustering}
\input{hardnessRgather}

\section{Subquadratic Hardness for the Discrete \frechet Distance}
\input{frechetHardness}

\section{Conclusion and Future Work}
An interesting direction of future work is to experiment with the time-window \frechet distance and the metric-based edit distance for real mobility trajectories. 

\newpage

\appendix
\section{Time-Window Metrics}
\input{appendixForTimeWindow}

\section{Metric-based Edit Distance}
\input{appendixForEditDistance}

\input{insertFirstEdit}

\section{\texorpdfstring{$k$}-gather Clustering}
\input{appendixForKGather}

\section{Subquadratic Hardness for the Discrete \frechet Distance}
\label{proof:parralel}
\input{appendixForDF}

\bibliography{frechet,privacy,newprivacy,r-gather,minhash,jie}
\end{document}

%% file: related.tex
\section{Related Work}


\paragraph{Trajectory Analysis}
The study of spatial temporal patterns that summarize the collective behaviors of moving objects has been a subject of study under the name of \emph{sequential pattern mining}, such as Generalized Sequential Patterns (GSP)~\cite{Srikant96miningsequential}, PrefixSpan~\cite{914830}, SPAM~\cite{Ayres:2002:SPM:775047.775109}, Convoy Pattern Mining~\cite{Jeung:2008wn,Jeung:2008fz}. 
There is also prior work that considers the temporal relations between events, or patterns that contain both spatial and temporal information~\cite{Giannotti:2007ha,Liu:2012gz}. Giannotti et al.~\cite{Giannotti:2007ha} handled GPS trajectories and identify popular regions and then connected them with the temporal patterns. Liu et al.~\cite{Liu:2012gz} focused on trajectory pattern mining in noisy RFID data. 
The mined frequent trajectories are used for prediction or classification of the trajectories~\cite{5560657}, labelling of locations~\cite{Zheng:2009:MIL:1526709.1526816}, or deriving mobile user behaviors~\cite{5560659}. 

\paragraph{Fr\'echet Distance}
There was a rich literature for \frechet distance. Alt and Godau first introduced the Fr\'echet distance in \cite{alt1995computing}, and showed that the Fr\'echet distance can be computed in $O(n^2\log n)$ time
for two curves with $O(n)$ vertices.
Eiter and Mannila \cite{eiter1994computing} used dynamic programming to compute the discrete Fr\'echet distance in $O(n^2)$ time. Since then there has been a sequence of results trying to beat the quadratic running time. 
In 2007 Buchin et al. \cite{buchin2007difficult} gave a lower bound of $\Omega(n \log n)$ for deciding whether the Fr\'echet distance between two curves is smaller than or equal to a given value.
In 2014 Agarwal et al. \cite{agarwal2014computing} proposed a subquadratic algorithm for computing the \emph{discrete} Fr\'echet distance between two curves in $O(n^2\log\log n / \log n)$.
Buchin et al. \cite{buchin2014four} provided a randomized algorithm to compute the Fr\'echet distance between two polygonal curves in time $O(n^2\sqrt{\log n}(\log \log n)^{3/2})$ on a pointer machine and in time $O(n^2(\log \log n)^2)$ on a word RAM.
In the same year, Bringmann \cite{bringmann2014walking} proved that the discrete Fr\'echet distance can not be computed in strongly subquadratic time, namely $O(n^{2-\varepsilon})$ for any $\varepsilon>0$, with the assumption of Strong Exponential Time Hypothesis (SETH). 
Later Bringmann and Mulzer \cite{bringmann2015approximability} showed that it is impossible to approximate the discrete \frechet distance with a factor smaller than $1.399$ in $\reals^1$ in strongly subquadratic time  with the assumption of SETH.

Meanwhile, there has been a series of near-linear approximation algorithm with different assumptions of input curves.
Alt et al. \cite{alt2004comparison} showed that the Fr\'echet distance between two $\kappa$-straight curves is bounded by 
$\kappa+1$ times the Hausdorff distance, leading to a $(\kappa+1)$-approximation by calculating Hausdorff distance in $O(n\log n)$. 
Aronov et al. \cite{aronov2006frechet} presented a near-linear-time $(1+\varepsilon)$-approximation algorithm for computing the discrete Fr\'echet distance between $\kappa$-bounded curves and backbone curves in two dimension.
Driemel et al. \cite{driemel2012approximating} provided $(1+\varepsilon)$-approximation algorithm for $c$-packed curves in 
$O(cn/\varepsilon + cn \log n)$.  
Gudmundsson et al. \cite{Gudmundsson:2018:FFD} provided an $O(n\log n)$ algorithm for computing \frechet distances of curves with very long edges -- much longer than the \frechet distance. 
There are also variants of \frechet distance, such as geodesic \frechet distance~\cite{wenk2010geodesic} and homotopic \frechet distance~\cite{chambers2010homotopic}. 

Of most relevance to this paper is work that considered speed or temporal constraints.
Notice that in computing \frechet distance one can move infinitely fast which is probably not practical. 
Maheshwari~\cite{MAHESHWARI2011110} examined \frechet distance when the moving speed on each segment is a constant within a given range. 
Buchin et al.\cite{buchin2010constrained} proposed a generic framework for constrained free space diagram. The constraints can be temporal, distance, directional, or attributional. 
Our time-windowed \frechet distance is a constrained type as well, applied to specific temporal requirements in matching data points in passively sensed trajectories.

\paragraph{Dynamic Time Warping}
Dynamic Time Warping (DTW) \cite{berndt1994using} is similar to \frechet distance in the sense that it also considers monotonic coupling of points on the two curves. But DTW takes sum of distances between coupled points, instead of the maximum as in \frechet.
Constraints on the match, such as Sakoe-Chimba band \cite{sakoe1978dynamic}, can be added to dynamic time warping to satisfy specific requirements. 

\paragraph{Edit Distance }  
A passively collected trajectory is represented by a sequence of checkpoints visited by a mobile entity. It is natural to use edit distance. 
The original edit distance assigns the same cost for operations of insertion, deletion and substitution, which is not suitable for trajectory analysis. 
From a geometric perspective, the cost of an operation over a far-away point should have a higher impact than a near point on the total cost. 
Therefore, variants of edit distance with different cost functions are applied to trajectory analysis. 
Chen et al. \cite{chen2004marriage} proposed edit distance with real penalty (ERP), where the cost of substitution between two points is the distance between them, and the cost of insertion and deletion of a point is the distance between that point and a constant reference point.  
Chen et al. \cite{chen2005robust} also proposed edit distance on real sequence (EDR), where the cost of substitution between two point within a threshold is $0$, and the cost of all the other operations is $1$. 
Yuan and Raubal \cite{yuan2014measuring} developed a spatio-temporal edit distance, whose cost of all operations is the distance between the new centroid of points after the operation and the previous centroid, where time is also considered as one dimension. 

For trajectory analysis, most existing approaches requires a metric, since the triangle inequality can help with pruning for efficient indexing and clustering.  In the above variants, ERP is a metric while EDR and spatio-temporal edit distance are not. The problem for ERP is that the ERP between two trajectories changes when the absolute position of them are changed even if the relative position remains,  because of using the constant reference point.   

\paragraph{Trajectory Clustering} Clustering trajectories has been studied in a number of prior work. Buchin\etal~\cite{buchin2008detecting} considered finding clusters defined as long sub-trajectories with small \frechet distance between them. They proved hardness and provided approximation algorithms. 
Buchin \etal~\cite{buchin2017clustering} also applied subtrajectory clustering into map construction. 
The use of $r$-gather for clustering was initially proposed by Aggarwal et al.~\cite{aggarwal2006achieving} and then applied for time-parameterized trajectories~\cite{Zeng:2017:MRD}. 
Agarwal et al.\cite{agarwal2018subtrajectory} also proved NP-hardness and provided fast approximation for finding the optimal collection of subtrajectory clusters that best represents the trajectories. 
Gudmundsson \etal~\cite{gudmundsson2013algorithms} provided algorithms to cluster trajectories for  hotspot finding problems.

%% file: frechetDefinition.tex
\section{Time-Window Metrics}
\subsection{\frechet Distance}

The \frechet distance is one of the most popular distance measure of two curves in space. It can be intuitively understood as a man traversing a finite curved path while walking his dog on a leash, with the dog traversing a separate path. Both the man and the dog need to walk forward but can take any speed at any time. The \frechet distance is the minimum length of the leash needed to finish the walk. An approximation and simpler version of the \frechet distance between polygonal curves is the discrete \frechet distance, which considers only positions of the leash whose endpoints are located at vertices of two polygonal curves. 
Now we start with the formal definitions.  

\begin{wrapfigure}{ro}{0.31\textwidth}
\vspace*{-6mm}
\centering
\subfigure[Two curves]{
\includegraphics[width=0.3\textwidth]{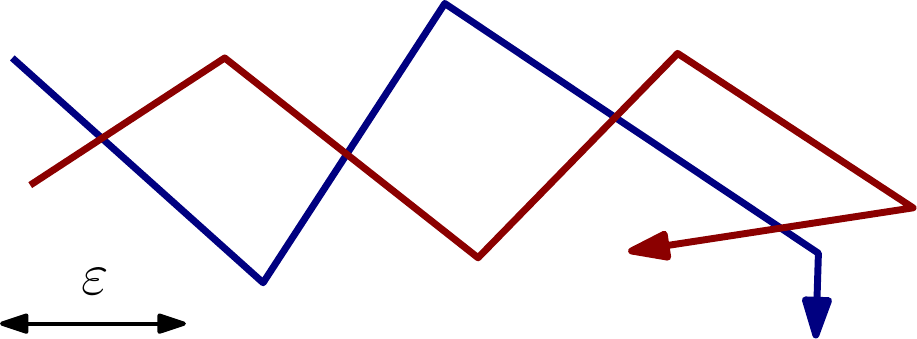}}\\
\subfigure[Free Space]{%
\label{fig:first}%
\includegraphics[width=0.3\textwidth]{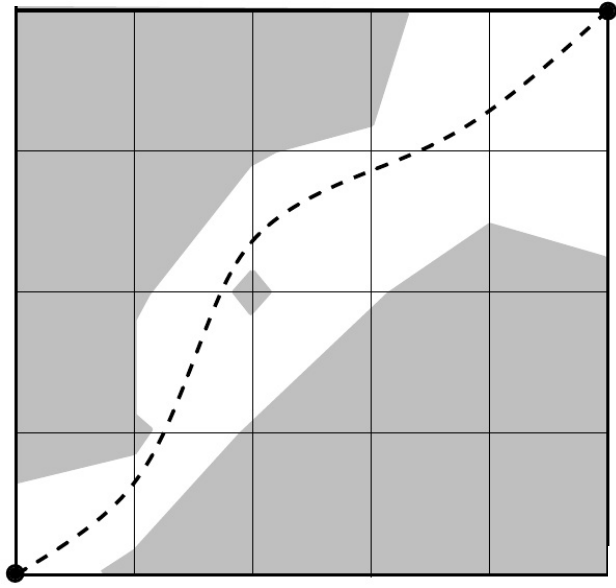}}
\vspace*{-2mm}
\caption{\small{Given $\varepsilon$, the white area in the left figure denotes the free space. The horizontal axis corresponds to the red curve and the vertical axis corresponds to the blue curve.} }
\label{Fig:freespace}
\vspace*{-4mm}
\end{wrapfigure}

A curve in $\reals^2$ is a continuous map $ [0,1]\rightarrow \reals^2$. A re-parameterization of $[0,1]$ is a continuous, non-decreasing, surjection $[0,1]\rightarrow [0,1]$. 
\begin{definition}[\frechet distance]
Let $A,B$ be two given curves in $\reals^2$, and $\alpha, \beta$ be re-parameterizations of $A,B$. The \frechet distance $\delta_F(A,B)$ is defined as
$$\delta_F(A,B) = \inf_{\alpha, \beta} \max_{s\in[0,1]} \left\lbrace d(A(\alpha(s)), B(\beta(s))) \right \rbrace , $$
where $d(\cdot, \cdot)$ is the distance between two points.  
\end{definition}

It is often easier to understand the \frechet distance between polygonal curves by the free space diagram. 
A polygonal curve $A$ is a continuous map $[0,n]\rightarrow \reals^2$, $A(i)=a_i\in \reals^2$ is the $i$th vertex.
The  map from interval $[i,i+1]$ to the $i$th line segment $A(i)A(i+1)$ of the curve is affine, $$A(s)=(1+i-s)a_i + (s-i)a_{i+1}, s\in [i,i+1].$$ 
We also represent $A$ as a sequence of its vertices $[a_0,a_1,\cdots, a_n]$. 

\begin{definition}[Free space~\cite{alt1995computing}] Free space between two polygonal curves $A$ and $B$ for a given distance $\varepsilon$ is defined as
$$D_{\leq\varepsilon}(A, B) = \left \lbrace (s, s') \in [0,n]\times[0,m]  \bigg | d(A(s),B(s'))\leq \varepsilon \right \rbrace .$$ 
\end{definition}

The free space is a two-dimensional region in the parameter space consisting of all points on the two curves with distance at most $\varepsilon$.  
A free space diagram is the free space along with vertical lines corresponding to vertices in $A$ and horizontal lines corresponding to vertices in $B$, as shown in Figure \ref{Fig:freespace}. 
These vertical and horizontal lines divide the free space into \emph{cells}. 
A cell $\mathcal{C}_{ij}=[i,i+1]\times [j, j+1]$ corresponds to all possible pairing between line segment $A(i)A(i+1)$ and $B(j)B(j+1)$.    
The \frechet distance between $A$ and $B$ is at most $\varepsilon$ if there is a path that is monotone in both $x$ (horizontal) and $y$ (vertical) dimension in the free space $D_{\leq\varepsilon}(A, B)$ from $(0,0)$ to $(n,m)$.  Given $\varepsilon$, we can decide whether there exists such a path in $O(nm)$ and we can find the \frechet distance using parametric search in $O(nm \log (nm))$~\cite{alt1995computing}. 


We can define the discrete \frechet distance on polygonal curves. 
\begin{definition}[Coupling/Traversal]
Given two polygonal curves represented by a sequence of vertices $A=[a_0,a_1,\dots, a_n], B =[b_0,b_1,\dots, b_m]$, we define a coupling $\beta=[c_0,c_1,\dots, c_l]$ as a sequence of pairs of points in $A$ and $B$, $c_k=(c_k[0], c_k[1])$, where $c_k[0]\in A, c_k[1]\in B$, such that 
 \begin{itemize}
 \item  $c_0=(a_0,b_0)$, $c_l=(a_n,b_m)$,
 \item if $c_k=(a_i,b_j)$, $c_{k+1}\in \{(a_i, b_{j+1}), (a_{i+1},b_j), (a_{i+1},b_{j+1})\}$.
 \end{itemize}
The distance of a coupling is the maximum distance among the pairs $c_k$ in $\beta$: $\max_{(a_i, b_j)\in \beta} d(a_i,b_j)$. The coupling is sometimes called traversal in the literature. 
\end{definition}

\begin{definition}[The discrete \frechet distance~\cite{eiter1994computing}]
The discrete Fr\'echet distance between $A$ and $B$ is defined as the minimum of  the maximum widths of all possible couplings of $A, B$: 
$$\delta_{dF}(A,B) = \min_\beta \max_{(a_i, b_j) \in \beta} d(a_i , b_j ) .$$ 
\end{definition}


It has been shown in \cite{eiter1994computing} that the discrete \frechet distance provides an upper bound for the \frechet distance and the difference between these two distances is bounded by the longest edge length of the polygonal curves. 
The discrete \frechet distance is sensitive to sampling rate of the curve. 
Usually, we do not want the sampling rate to affect our measure. So we focus on the continuous \frechet in the algorithms of this paper. At the very end we present results on the hardness of approximation for the discrete \frechet distance.

Both discrete and continuous \frechet distances are measures of \emph{shapes}, spatial information, of curves and do not take temporal information that could be associated with a trajectory. 

%% file: timewindow.tex
\subsection{Time-Window Fr\'echet Distance} 

In this section, we introduce time window \frechet distance to analyze human mobility trajectories. 
In the original \frechet distance, there are no constraints on the couplings with which we obtain the \frechet distance. This loses important information in the temporal dimension for comparing time-stamped trajectories. 

A trajectory is a continuous map from time $t\in [0,1]$ to $\reals^2$:  $[0,1]\rightarrow \reals^2$. 
W.l.o.g., we normalize time into the range $[0,1]$ and assume all trajectories start at time $0$ and end at time $1$. 
\begin{definition}[Time window \frechet distance]
Let $A, B$ be  two trajectories, and
$\alpha, \beta$ be the re-parameterization of $A, B$.  
The time window \frechet distance is defined as 
$$\delta_{tF}^{\sigma}(A, B) = \inf_{\alpha, \beta} \max_{s\in[0,1]} \left\lbrace d(A(\alpha(s)), B(\beta(s))) \right \rbrace$$ 
for any $s\in [0,1]$, we have $\left | \alpha(s)-\beta(s) \right | <\sigma$, where $\sigma$ is a parameter specifying which points on two trajectories can be paired. 
\end{definition}
Since we add constraints on the original \frechet distance, a valid re-parametrization for time window \frechet distance is a valid coupling for \frechet distance. Hence, we have the following lemma. 
\begin{lemma}
\begin{enumerate}
\item $\delta_{tF}^{\sigma}(A, B) \geq \delta_{F}(A,B).$
\item $\delta_{tF}^{\sigma}(A, B) \geq \delta_{tF}^{\sigma'}(A,B)$, if $\sigma \leq \sigma'$.
\end{enumerate}
\end{lemma}

Trajectories in the real world settings are often produced by discrete samples taken by various localization techniques. 
Without further information, we assume that the mobile entities move along the polygonal curves determined by the sample points. In this section we consider two scenarios:
\begin{enumerate}
\item \textbf{Constant Speed.} A mobile entity travels at constant speed between two consecutive sample points, \label{const}
\item \textbf{Varying Speed.} A mobile entity may travel at varying speed from a sample point to the next.  \label{arbit}
\end{enumerate}

\begin{figure}[t!]
\centering
\subfigure[The constraints on free space diagram of trajectory $A,B$ under constant speed assumption]{
\label{Fig:constant}
\includegraphics[width=0.35\textwidth]{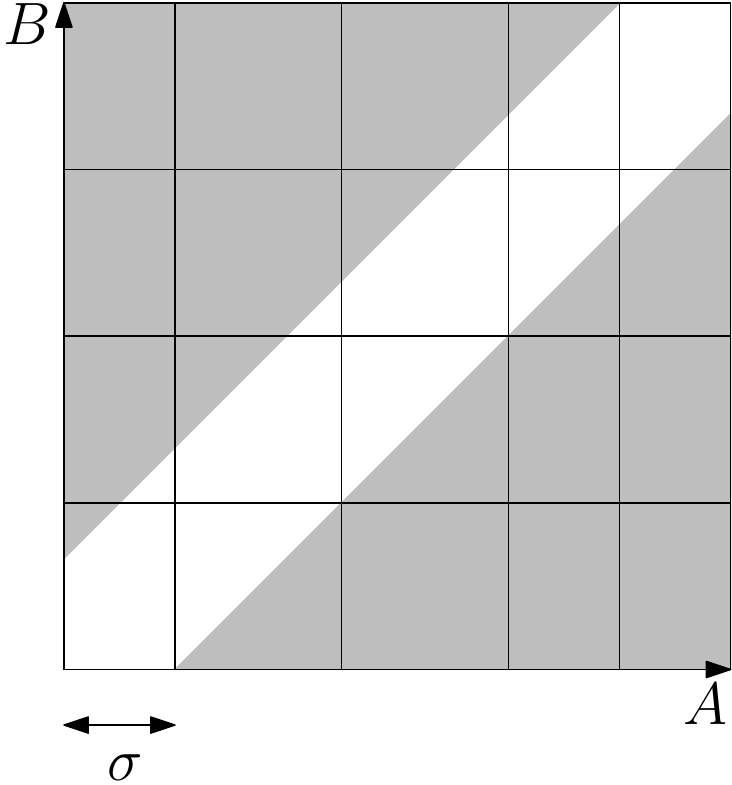}}
\hspace{5mm}
\subfigure[The constraints on free space diagram of polygonal curve $A',B'$ under arbitrary speed assumption]{
\label{Fig:arbitrary}
\includegraphics[width=0.35\textwidth]{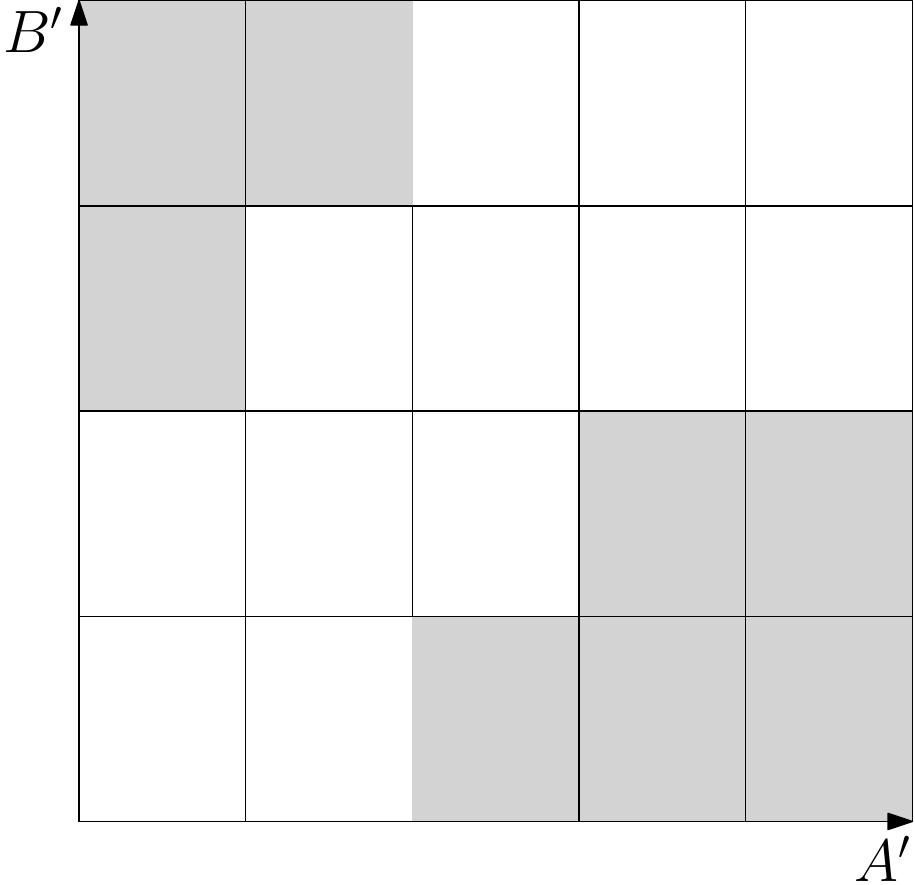}}
\caption{Constraints on the free space diagram for time window \frechet distance given $\sigma$.  }
\label{Fig:timewindowFreespace}
\end{figure}


Given two trajectories $A,B$ with discrete sample points, $A(t_0) = a_0, A(t_1)=a_1,\cdots, A(t_n)=a_n, 
B(t_0')=b_0,B(t_1')=b_1,\cdots, B(t_m')=b_m$, $a_i,b_j \in \reals^2$,
we first show the difference of the time window \frechet distance under constant speed and varying speed conditions and then present algorithms.

\smallskip\noindent\textbf{Constant Speed. }
If a mobile entity travels at constant speed between two consecutive sample points, we can run interpolation to get the location of the mobile entity at time $t$.

In the free space diagram, we simply use time as axes. 
Use $t$ to refer time used for trajectory $A$ and $t'$ be time of $B$. 
All pairs of $t,t'$ within time window $\sigma$, $|t-t'|\leq \sigma$, is demonstrated in Figure \ref{Fig:constant}, where the horizontal axis is $t$ of $A$ and the vertical axis is $t'$ of $B$. 
A valid re-parametrization with time window constraints, is represented as a monotone path within the white area decided by $|t-t'|\leq \sigma$ in the free space diagram. 
Hence, we have the following lemma. 

\begin{lemma}
Given time window $\sigma$, 
$\delta^\sigma_{tF}(A, B)\leq \varepsilon$  iff there is a monotone path from $(0,0)$ to $(1,1)$ within the space decided by $|t-t'|\leq \sigma$ in the free space diagram $D_{\leq\varepsilon}(A, B)$, where $t,t'$ are time of $A$ and $B$. 
\end{lemma}

\smallskip\noindent\textbf{Varying Speed. }
If we allow a mobile entity travel at varying speed in the direction from a sample point to the next, the problem boils down to setting constraints only on how sample points of $A, B$ are paired. The set of sample locations of $B$, which $A(t_i)$ can be paired with, contains those that are within time interval $[t_i-\sigma, t_i+\sigma]$, and the first one outside this interval:
\begin{equation}\label{eq:varyingspeed}
\{ B(t_j') \big | |t_i -t_j'|\leq \sigma \text{ or } (t'_j<t_i-\sigma, t_{j+1}'\geq t_i-\sigma) \text{ or } (t_j'>t_i+\sigma, t_{j-1}'\leq t_i+\sigma) \}.
\end{equation}
Specifically, since we allow a mobile node travel with arbitrary speed, the first location sample of $B$ \emph{outside} the time window $\sigma$ from $t_i$ can also be paired with $A(t_i)$. Suppose we match $A(t_j)$ with $B(t_j')$ for $t'_j<t_i-\sigma, t_{j+1}'\geq t_i-\sigma$. This can happen if the mobile entity traveling on $B$ stays at the location $B(t_j')$ till time $t_i-\sigma$. 
Similarly,  $B(t_j')$ for 
$t_j'>t_i+\sigma, t_{j-1}'\leq t_i+\sigma$ can be paired with $A(t_i)$, if the mobile entity on $A$ stays at $A(t_i)$ till $t_j'-\sigma$. 
From the above analysis, each $A(t_i)$ can be paired with at least two locations of $B$. 
Such pairing is monotone in time, that is, if  $A(t_i)$ can be paired with  $B(t_j'),B(t_{j+1}'),\cdots, B(t_k') $ and  $A(t_{i+1})$ can be paired with  $B(t_{j'}'),B(t_{j'+1}'),\cdots, B(t_{k'}') $, we have $t_j'\leq t_{j'}'$, $t_k'\leq t_{k'}'$. 

Once a pairing of vertices of $A$ and $B$ satisfy the constraints above, there is a possible assignment of velocities to the interior points on a segment such that the pairing are within time window of $\sigma$ of each other.

In this setting, the map in a trajectory function from time to the line segment between two sample points is not linear, so we transform trajectory $A,B$ into their polygonal curve parametrization, $A':[0,n]\rightarrow\reals^2, B':[0,m]\rightarrow \reals^2.$ 
In the free space diagram, we make the horizontal axis represent domain of $A'$ and the vertical axis represent domain of $B'$. A point on the $x$-axis of the free space diagram represents a point on $A'$.
If both $A'(i)$ and $A'(i+1)$ can be paired with $B'(j)$ and $B'(j+1)$ by Condition~\ref{eq:varyingspeed}, there exists a way in which two mobile entities travel on line segment $A'(i)A'(i+1)$ and $B'(j)B'(j+1)$ during $[t_i, t_{i+1}]$ and $[t_j', t_{j+1}']$, respectively, s.t. 
for any 
pairing of the two mobile entities' locations on a re-parametrization in the cell $\mathcal{C}_{ij}=[i,i+1]\times[j,j+1]$ on free space diagram, the difference of time when the two mobile entities visiting the locations, is within time window $\sigma$. 
Then, we say that a cell $\mathcal{C}_{ij}$ is valid under the time window constraint. 
See Figure \ref{Fig:arbitrary} for an example. 
Hence, we have the following lemma. 

\begin{lemma}
Given time window $\sigma$, 
$\delta^\sigma_{tF}(A, B)\leq \varepsilon$  iff there is a monotone path from $(0,0)$ to $(n,m)$ within the union of all valid cells in $D_{\leq\varepsilon}(A, B)$. 
\end{lemma}

\smallskip\noindent\textbf{Algorithm and Running Time Analysis. } 
We use the free space diagram to find our time-window \frechet distance, 
where we enforce all monotone paths to stay within the region defined by the time window. 
Now, we focus on analysing the \frechet distance under the above two assumptions. 
We denote by $C(n,m,\sigma)$ the number of cells containing space satisfying time window constraint $\sigma$ in both above assumptions of constant and varying speeds.  
To find all such cells, since points on trajectories are already in chronological order, we can do a linear scan over trajectories in $O(n+m)+C(n,m,\sigma)$. 
Given $\varepsilon$, we can check whether $\delta^\sigma_{tF}(A,B)\leq \varepsilon$ in time $C(n, m, \sigma)$,  similar to \cite{buchin2010constrained}.  
With  Cole's parametric search~\cite{Cole1987Slowing} in  \cite{alt1995computing}, the complexity of finding the time window \frechet distance is $O(C(n,m,\sigma) \log C(n, m, \sigma))$. 
Since for each sample point on the trajectories, there is at least one corresponding cell containing the space satisfying the time window constraints, $C(n,m,\sigma)\geq n$. All above, the overall complexity is $O(C(n,m,\sigma)\log C(n,m,\sigma))$. 

\begin{remark}
We can also apply the time-window constraint on discrete \frechet distance and dynamic time warping, as shown in Appendix \ref{sec:othertimewindow}. 
\end{remark}

%% file: weight.tex
\section{String Representation and Metric-Based Edit Distance}

For  passively collected trajectories, 
when a mobile entity get connected to or approximated to a labeled checkpoint, such as cellular towers, WiFi Access Points, or other stations, the appearance of the mobile entity is collected. 
A trajectory is represented as a sequence of checkpoints visited in order. 
If we associate each checkpoint by a unique  character, we get a \emph{string representation} of the mobile entity's trajectory. 
Such trajectories contain rich information and save tremendous storage compared with trajectories of frequent GPS sample points.

\subsection{Edit Distance and Metric-Based Edit Distance}
With the string representation of a trajectory, a natural metric is the edit distance. The distance between two strings $A = a_1a_2\cdots a_n$ and $B = b_1b_2\cdots b_m$ is the smallest number of character insertions and deletions that convert $A$ to $B$. One can compute the edit distance by using dynamic programming in time $O(nm)$.

\def\ins{\operatorname{INS}}
\def\del{\operatorname{DEL}}
\def\cost{\operatorname{D}}
\def\back{\operatorname{B}}
\def\DP{\operatorname{D}}
\def\dist{d}
For trajectories, it would make sense to differentiate the insertion of a nearby location versus the insertion of a far away location. 
The definition of `nearby' or `far-away' location can be application dependent. For example, one can examine the Euclidean distance between two locations. Alternatively, it might be interesting to define such distance by the functionality (e.g., by the district partitioning in a city, by the type of buildings which a location belongs to, etc). We assume that this distance is defined by the metric $\dist(\cdot, \cdot)$. 

Given $\dist$, we propose a \emph{metric-based edit distance} with insertion and deletion cost as following.
Let $U$ be the set of all characters representing locations.
For $x,y,z \in U$,
we define the insertion cost to insert $z$ between $x$ and $y$, as the difference of taking the detour through $z$ rather than going straight:
\begin{equation}
\ins(x, y, z) = \dist(x,z) + \dist(y,z) -\dist(x,y). 
\end{equation}
Similarly, we define the deletion cost to delete $z$ between $x$ and $y$ symmetrically:
\begin{equation}
\del(x, y, z) =  \dist(x,z) + \dist(y,z) -\dist(x,y)
\end{equation}

To define the insertion and deletion cost before the first character and after the last character of a string, we add  character $S$ and $T$ to the beginning and end of each trajectory string. Here $S, T$ are dummy nodes with distance $M=\infty$ to all other locations.  
If $z$ stays on the line segment between $x, z$, $\ins(x,y,z)=0$, $\del(x,y,z)=0$. This makes sense as the shape of the trajectory does not change and the insertion/deletion of $y$ properly handles possible missing data in this case. 
Besides, the cost of inserting or deleting a sequnce of $y$ between $x$ and $z$ has the same cost with inserting or deleting one $y$ between $x$ and $z$.  In this case, redundant data is also handled. 

\begin{wrapfigure}{r}{0.43\textwidth}
\begin{center}
\vspace*{-14mm}
\begin{tabular}{cc}
	\includegraphics[scale=.4]{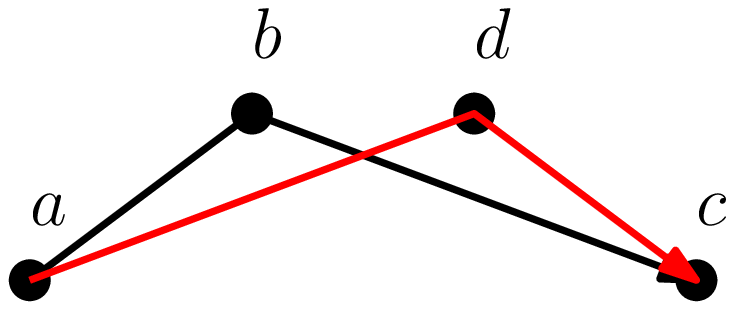} &
	\includegraphics[scale=.67]{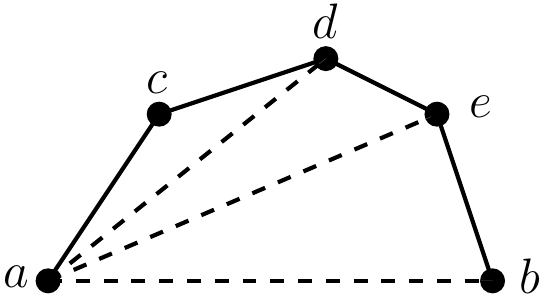} \\
\footnotesize (a) & \footnotesize (b) 
\end{tabular}
\end{center}
\vspace*{-4mm}
	\caption{\footnotesize \textbf{(a)}  The metric-based edit distance between trajectory $abc$ and $adc$ is $\min\{|ab|+2|bd|+|dc|-|ad|-|bc|, |ab|+|bc|+|ad|+|dc|-2|ac|\}$. The former is the cost of inserting $d$ first and then deleting $b$, while the latter is the cost of deleting $b$ first and then inserting $d$. 
  \textbf{(b)} The cost of deleting $cde$ between $a$ and $b$ is always $|ac|+|cd|+|de|+|eb|-|ab|$ for all the orders to delete $cde$.  }
	\label{fig:dist2}
\vspace{-6mm}
\end{wrapfigure}

\begin{definition}[Metric-based edit distance]
Given two trajectories with string representation, the \emph{metric-based edit distance} is the minimum cost to convert one trajectory string to the other with arbitrary order of insertions and deletions, where the cost of insertion and deletion is defined by function $\ins$ and $\del$. 
\end{definition}


In some edit distance definitions, one can substitute one character by another. Here we use insertion and deletion to simulate a substitution operation. 
An example of metric-based edit distance is shown in Figure \ref{fig:dist2}(a).


%

\begin{lemma}
\label{lem:insertcost}
Given two characters in a trajectory string, 
  the total cost of deleting all characters between these two characters is not affected by the order of performing these deletions. Given two neighboring characters in a trajectory string, the total cost of inserting a sequence of characters between these two characters is also not affected by the order. 
\end{lemma}
Proof see Appendix \ref{proof:insertcost}. An example is in Figure \ref{fig:dist2}(b). 
\begin{theorem}
\label{thm:metric}
The metric-based edit distance is a metric.
\end{theorem}
Proof see Appendix \ref{proof:metric}.

\subsection{Algorithm for Metric-based Edit Distance}
The insertion and deletion cost of  metric-based edit distance is decided by the neighbors. This provides us a cost highly related to the underlying distance metric of two symbols. On the other hand, it also increases the complexity of computing the distance as we need to consider different orders of insertion and deletion of symbols as these may affect the cost in later insertion and deletion operations. 
The good thing is that 
we can still run dynamic programming algorithm with running time $O(n^3m^3(n+m))$ with details provided in Appendix \ref{sec:arbitrary}. 

We can simplify the computation if we require all insertions be done before any deletions when converting $A$ to $B$. In this way the running time can be made to be
in $O(nm)$. But the downside is that the distance no longer satisfies the triangle inequality.
See Appendix \ref{sec:insertFirst}.


%% file: hardnessRgather.tex

A common practice to process trajectories is to perform clustering. That is, trajectories that are similar to each other are grouped in one cluster. Tight and dense clusters of trajectories naturally correspond to meaningful features such as group motion, convoy, etc. A proper notion for this purpose is the $k$-gather problem, which requires each cluster to have at least $k$ trajectories. 

\begin{definition}[$k$-gather problem\cite{aggarwal2006achieving}]
The $k$-gather problem is to cluster $n$
points in a metric space into a set of clusters, such that each
cluster has one point as the center and at least $k$ points. The objective is to minimize the
maximum radius among the clusters, where the radius is the distance from a point in a cluster to the center of the cluster. 
\end{definition}


The $k$-gather problem on a general metric distance is NP-hard to compute when $k>6$~\cite{aggarwal2006achieving}. We show that for our specific metric distances between trajectories the problem remains hard. Our gadgets construction was motivated by the original proof \cite{aggarwal2006achieving} and need to be carefully created to fit the trajectory setting.





\begin{theorem}
	\label{thm:hardedit}
The $k$-gather problem of trajectories on edit distance and metric-based edit distance is NP-hard, for $k>13$.
\end{theorem}

The reduction is from 3SAT. We create a set of trajectories such that the radius of the edit distance (or metric-based edit distance) is $5$ if and only if the 3SAT instance is satisfied.
We define a planar graph with each face representing a character in the string representation of trajectories. The faces are clustered into $4n$ cells $c_1, c_2,\cdots, c_{4n}$. The first $3n$ cells each consists of $3$ faces each: $2$ rectangle faces side by side and $1$ triangle face on top of them, while for the last $n$ cells each consists of $2$ rectangle faces side by side and $2$ triangle faces adjacent to both rectangle faces as shown in Figure \ref{fig:NPproof1}. 
\begin{figure*}[htpb]
	\centering
	\includegraphics[width=\textwidth]{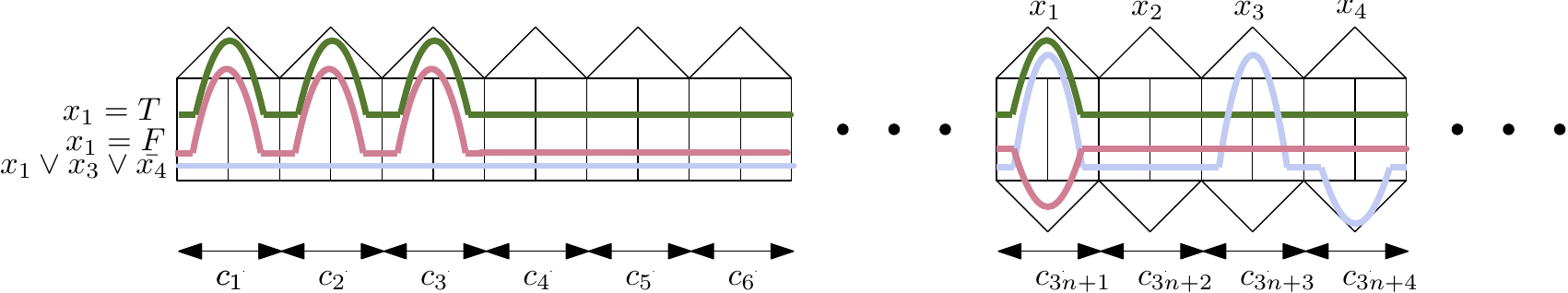}
	\caption{The variable gadget to reduce from $3$SAT problem.}
	\label{fig:NPproof1}
\end{figure*}
	
All trajectories go from $c_1$ to $c_{4n}$ through all the rectangular faces. But some trajectories take `detours' to visit some of triangular faces. We call that a top detour (visiting the face above the rectangular faces) or a bottom detour (visiting the face below the rectangular faces). For the part of trajectories within each cell, the edit distance between a trajectory with a detour and a trajectory going straight is $1$, and the edit distance between a trajectory with a top detour and  one with a bottom detour is $2$.
	
For each variable $x_i$, we construct two trajectories for variable assignment $x_i=\mbox{True}$ and $x_i=\mbox{False}$ respectively. Both trajectories have top detours in $c_{3i-2}, c_{3i-1}, c_{3i}$; the trajectory of $x_i=\mbox{True}$ has a top detour in cell $c_{ 3n+i}$ while the trajectory of $x_i=\mbox{False}$ has a bottom detour in that cell; and there are no detours for all the other cells in the two trajectories. 
	For each clause, we create $3$ trajectories such that for each variable $x_i$ in the clause, each of the three trajectories has a top detour in cell $c_{3n+i}$ , else if $\overbar{x_i}$ is in the clause, each of the three trajectories has a bottom detour in cell $c_{3n+i}$; the three trajectories have no detours in the rest cells. 

If the 3SAT instance is satisfiable, there is a solution for $k$-gather by picking the trajectory representing the assignment as the center of a cluster. By playing with the edit distance/metric-based edit distance and additional filler trajectories (to make up $k$ trajectories in a cluster), we get the results. 
For the full proof see Appendix \ref{sec:hardedit}.

We also prove the hardness of $k$-gather on Jaccard distance in Appendix \ref{sec:hardnessJacc}. 
For completeness, we provide a $2$-approximation algorithm for the $k$-gather problem on these metrics based on the work~\cite{aggarwal2006achieving} in Appendix \ref{sec:approxAlgo}.

%% file: frechetHardness.tex
\begin{wrapfigure}{r}{0.4\textwidth}
\centering
\vspace*{0mm}
\includegraphics[width=0.4\textwidth]{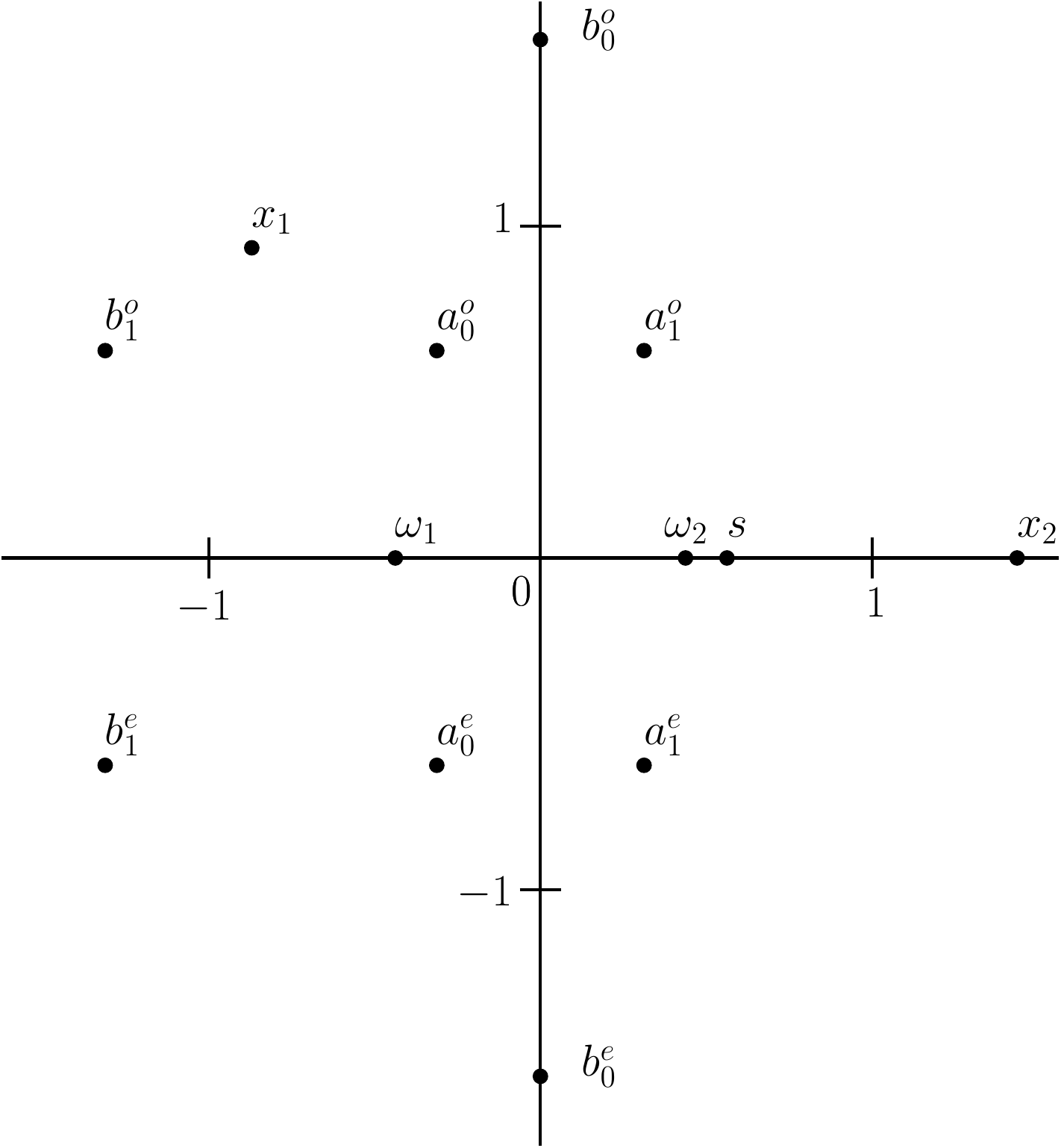}
\vspace*{-6mm}
\caption{The points for $A_i,B_j$. }
\label{Fig:approxHardness}
\vspace*{-15mm}
\end{wrapfigure}

The discrete \frechet distance can be computed with dynamic programming in $O(n^2)$~\cite{eiter1994computing}. 
Bringmann in~\cite{bringmann2014walking} proved that strongly subquadratic algorithms for the discrete \frechet distance are unlikely
to exist, unless the Strong Exponential Time
Hypothesis (SETH) fails, even in the one-dimensional case. Bringmann and Mulzer proved that it is true even if we allow an approximation up
to a factor of $1.399$~\cite{bringmann2015approximability}. We examined the construction of \cite{bringmann2015approximability} and we show in $\reals^2$
subquadratic algorithms are unlikely to exist for approximation factors up to $1.61$ for $\ell_2$ norm. 
Similar to prior work the reduction is from the Orthogonal Vector (OV) problem.


\begin{definition}[Orthogonal Vectors (OV) Problem]
Assume we have two collections of $D$-dimensional binary vectors $u_1,\dots, u_N$ and $v_1,\dots, v_N$, where $u_i, v_j \in \{0, 1\}^D$. We want to decide whether there is a pair of vectors $u_i$, $v_j$, such that $u_i\perp v_j$.
\end{definition}

\begin{lemma}[\cite{williams2004new}]
If there exists an $\eps>0$ such that there is an algorithm to solve the Orthogonal Vectors problem OV with running time $D\cdot N^{2-\eps}$, then SETH fails. 
\end{lemma}

Let $u_1,\dots, u_N$ and $v_1,\dots, v_N$ be the sets of vectors in the Orthogonal Vectors problem. Without loss of generality, we assume the dimension these vectors, $D$, is even -- otherwise we can pad an extra dimension with all values being $0$. We show how to construct two sequences $P$ and $Q$ in time $D\cdot N$ such that there are two orthogonal vectors $u_i$, $v_j$, $u_i \perp v_j$ if and only if $\delta_{dF}(P, Q) \leq 1$.


The sequences $P, Q$ consist of points as shown in Figure \ref{Fig:approxHardness}, where 
\begin{equation*}
\begin{split}
b_0^o=(0, 1.61),   &\, b_0^e=(0, -1.61), \\
b_1^o=(-1.305,0.66), & \, b_1^e=(-1.305,-0.66), \\ a_0^o=(-0.305,0.66), & \, a_0^e=(-0.305,-0.66),\\ a_1^o=(0.305,0.66), & \, a_1^e=(0.305,-0.66), \\s=(0.555,0), & \\
\omega_1=(-0.445,0),& \, \omega_2=(0.445,0), \\
x_1=(-0.88,0.90), & \, x_2=(1.445,0). 
\end{split}
\end{equation*}

%

Other than $x_1, x_2, s$, the other points labeled with $o$ and $e$ are symmetric with respect to the $x$-axis.

We first construct the vector gadgets. For  
$u_i=[u_{i,1},u_{i,2},\cdots, u_{i,D}]$, 
we construct a sequence of $D$ points, denoted by 
$A_i=[A_{i,1}, A_{i,2},\cdots, A_{i,D}]$.  
For the $k$th element $u_{i,k}$, the corresponding
point  $A_{i,k}$ is $a_{u_{i,k}}^{p}$ where $p$ is the parity of $k$, $p=o$ if $k$ is odd and $p=e$ if $k$ is even. In the same way,  for each $v_j$, we construct a corresponding sequence $B_j$  with points $b_{\{0,1\}}^{\{o,e\}}$. 

\begin{definition}[Parallel coupling]
A coupling $\beta$ of $A_i, B_j$ is 
a parallel coupling, if $\beta =[(A_{i,1}, B_{j,1}), (A_{i,2}, B_{j,2}),\cdots, (A_{i,k}, B_{j,k}), \cdots, (A_{i,D}, B_{j,D})] $.  
\end{definition}


The proofs of the following Lemmas can be found in the Appendix \ref{proof:parralel}.
\begin{lemma}
\label{lem:parralel}
Given a coupling $\beta$ of $A_i, B_j$, $\delta(\beta)=\max_{(A_{i,k},B_{j,k'})\in\beta} d( A_{i,k}, B_{j,k'})$, where $d(\cdot, \cdot)$ is the Euclidean distance between two points. 
If $\beta$ is not a parallel coupling, then $\delta(\beta)> 1.61$.    
If $\beta$ is a parallel coupling, for $u_i \perp v_j$, $\delta(\beta)\leq 1$, and for $u_i \not\perp v_j$, $\delta(\beta)\geq 1.61$.
\end{lemma}

Let $W$ be a sequence of $D(N-1)$ points that alternates between $a_0^o$ and $a_0^e$. We set 
$$P=W\circ x_1 \circ \left( \bigcirc_{i=1}^N s \circ A_i \right) \circ s \circ x_2 \circ W, \; \text{ and } Q= \bigcirc_{j=1}^N \omega_1 \circ B_j \circ \omega_2, $$
 where $\circ$ denotes the concatenation of the sequences.

The subsequence $W$ is a buffer to make the \frechet distance  between $W$  and at most $N-1$ sequences $\omega_1 \circ B_j \circ \omega_2$ be within $1$. 
$\omega_1$ and $\omega_2$ on $P$ are points to make the \frechet distance between $\omega_1$, $\omega_2$ and sequences $s \circ A_i$ be within $1$. 
$s$ is used to synchronize in the coupling such that if the \frechet distance is at most $1$, the coupling between a pair of $A_i$ and $B_j$ must be parallel. 
$x_1, x_2$ guarantee at least one parallel coupling between a sequence for $u_i$ and a sequence $v_j$ exists if the distance of the coupling is at most $1$. 

Below we will prove a number of Lemmas for the main result with proofs in Appendix \ref{proof:dF1}, \ref{proof:perp}. First, we show that  if there exists a pair of orthogonal vectors, we  guarantee that the \frechet distance between the corresponding sequences of $P$ and $Q$ is no greater than $1$; also, if there are no orthogonal vectors, the distance between corresponding $P$ and $Q$ is greater than $1.61$.

\begin{lemma}
\label{lem:dF1}
If there exists a pair $u_i,v_j$ such that $u_i \perp v_j$, $\delta_{dF}(P,Q)\leq 1$. 
\end{lemma}

\begin{lemma}
\label{lem:perp}
If there is no $i,j$, such that $u_i \perp v_j$, $\delta_{dF}(P, Q)\geq 1.61$.
\end{lemma}

With the lemmas above, we can have the theorem below. 
\begin{theorem}
There is no strongly subquadratic approximation algorithm to compute the discrete \frechet distance with an approximation factor less than $1.61$ unless SETH fails. 
\end{theorem}

%% file: appendixForTimeWindow.tex
\subsection{Other Time-Window Metrics}
\label{sec:othertimewindow}
The idea of enforcing a time-window constraint can be applied to other distances such as the  discrete \frechet distance~\cite{eiter1994computing} and dynamic time warping~\cite{berndt1994using}. 
For polygonal curve $A=[a_0, a_1,a_2,\cdots, a_n]$ and $B=[b_0,b_1,b_2,\cdots, b_m]$, given distance metric $d(\cdot,\cdot)$, these two distances can be computed by dynamic programming in time $O(mn)$ as follows. 
\def\c{\operatorname{D}}
\begin{equation}
\begin{aligned}[l]
&\text{Dynamic time warp:} \\ 
&\c[i][j] = d(A[i], B[j]) + \min\{\c[i][j-1],\c[i-1][j-1], \c[i-1][j]\}
\end{aligned}
\end{equation}
\begin{equation}
\begin{aligned}[l]
&\text{Discrete \frechet distance: }\\
&\c[i][j] = \max\{d(A[i],B[j]), \min\{\c[i][j-1],\c[i-1][j-1], \c[i-1][j]\}\}
\end{aligned}
\end{equation}


We can apply our time window with varying speed to these two distances. 
Assume in trajectory $A,B$, $A(t_i)=a_i, B(t_j') = b_j$. Given the time window $\sigma$, we need only compute $\c[i][j]$ for $i,j$ satisfying 
$$ |t_i -t_j'|\leq \sigma \text{ or } (t'_j<t_i-\sigma, t_{j+1}'\geq t_i-\sigma) \text{ or } (t_j'>t_i+\sigma, t_{j-1}'\leq t_i+\sigma) .$$
With similar analysis in previous subsection,  the complexity is $\max\{O(n\log n), C(n,m,\sigma)\}$, where $C(n,m,\sigma)$ is the number of valid $i,j$ pairs.

Previously, Sakoe-Chimba band~\cite{sakoe1978dynamic} has been introduced to dynamic time warping, in which given a window parameter $w$, only values of $\c[i][j]$ for $i,j$ satisfying $|i-j|\leq w$, are computed.  This works only for trajectories with uniform sampling. 

%% file: appendixForEditDistance.tex
\subsection{Proof of Lemma \ref{lem:insertcost}}
\label{proof:insertcost}
\begin{proof}
The proof is by induction. 
Assume there are $\ell$ characters to delete between two characters in the string.
If $\ell=1$, the deletion cost is not affected by the order. 
If it is true with  the deletion of $\ell \leq k$ characters, for $\ell=k+1$, given a substring $z_0 z_1 \cdots z_{k+2}$, we  prove that if we delete all the $k+1$ characters between $z_0$ and $z_{k+2}$, the deletion cost is not affected by the order. 
If $z_{k+1}$ is the last deleted, we first delete  $z_1,z_2,\cdots,z_k$ between $z_0$ and $z_{k+1}$, whose cost is not affected by the order, and thereafter, delete $z_{k+1}$. 
 The total deletion cost is  
\begin{equation*}
  \begin{aligned}
&\sum_{j=1}^k \del(z_0,z_{j+1},z_j) + \del(z_0,z_{k+2},z_{k+1}) =& \sum_{j=0}^{k+1}\dist(z_j,z_{j+1})-\dist(z_0,z_{k+2}). 
\end{aligned}
\end{equation*}
The deletion cost is the  same when $z_1$ is the last deleted.
 
If $z_i$ is the last deleted for $1<i<k+1$, then $z_i$  divide the characters to delete into two parts $z_1,z_2,\cdots,z_{i-1}$ and $z_{i+1},\cdots,z_{k+1}$, both no greater than $k$. Hence, the deletion cost of both substring are not affected by the order. The total deletion cost is
\begin{equation*}
  \begin{aligned}
    & \sum_{j=1}^{i-1}\del(z_0, z_{j+1}, z_j) + \sum_{j=i+1}^{k}\del(z_{i},z_{j+1}, z_{j}) + \del(z_i,z_{k+2},z_{k+1}) + \del(z_0,z_{k+2},z_i)\\
  =&  \sum_{j=0}^{k+1}\dist(z_j,z_{j+1})-\dist(z_0,z_{k+2}). 
\end{aligned}
\end{equation*}
Hence, for $\ell=k+1$ the cost of deletions is not affected by the order. 
All above, we prove for any substring to delete, the deletion cost is not affected by the order. Similarly, we can prove the insertion cost of a substring between two neighboring characters is not affected by the order. 
\end{proof} 

\subsection{Proof for Theorem \ref{thm:metric}}
\label{proof:metric}
\begin{proof}
First, the metric-based edit distance is nonnegative, and equals to $0$ if and only if two trajectories have the same shape.  

The metric-based edit distance is symmetric. Assume we have a sequence of insertions and deletions to convert a trajectory string $A$ to a $B$ with minimum cost $\gamma$. We obtain a sequence of operations to convert $B$ to $A$, if we reverse the sequence of operations, change each deletion to an insertion, and change each insertion to a deletion.  The cost is also $\gamma$, since the cost of inserting $z$ between $x$ and $y$ is the same with the cost of deleting $z$ from $x$ and $y$. $\gamma$ is the minimum cost to convert $B$ to $A$. Otherwise, we can find a cost lower than $\gamma$ to convert $A$ to $B$ repeating the above reversing sequence procedure. 
Therefore, the metric-based edit distance is symmetric. 


The metric-based edit distance satisfies triangle inequality. If we insert one other character and delete it anytime in a procedure of converting trajectory string $A$ to $B$, the cost is no less than the metric-based distance between $A$ and $B$. This result can be generalized that the minimum cost $\theta$ of converting $A$ to $B$ during which we also insert and delete characters of a trajectory string $C$ is no less than the metric-based distance between $A$ and $B$.  The metric-based edit distance between $A$ and $C$ plus the one between $C$ and $B$ is no less than $\theta$, and hence, no less than the metric-based distance between $A$ and $B$. 
\end{proof}

\subsection{Algorithm for Metric-Based Edit Distance}
\label{sec:arbitrary}
Given two trajectories with string representation  $A=a_0a_1\dots a_na_{n+1}$ and $B=b_0b_1\dots b_mb_{m+1}$, where $a_0=b_0=S, a_{n+1}=b_{m+1}=T$, we calculate the minimum cost converting $B$ to $A$.   
$\cost(k, \ell, i, j)$ denote the minimum cost to transform from a trajectory substring $b_k b_{k+1} \cdots b_{\ell-1} b_\ell$ to $b_k  a_i a_{i+1} \cdots a_j b_\ell$ with arbitrary order of insertions and deletions.  
In the intermediate  string, all characters belonging to $A$ remain their relative order in $A$, otherwise the cost is greater.
Notice for arbitrary order $\cost(k, \ell, i, j)$ is not necessarily smaller than $\cost(k, \ell, i, j+1)$ and $\cost(k, \ell+1, i, j)$. 
We can pre-compute insertion and deletion of any sequence between any two elements in $A$ and $B$  in $O(n^2m^2)$, and we assume in the following analysis, the cost of any deletion or insertion of a sequence of characters can be obtained in $O(1)$ by checking the values pre-computed. 
Now we provide a dynamic programming algorithm to calculate the metric-based edit distance with arbitrary order of insertions and deletions.   
\begin{equation*}
\begin{aligned}
\cost(k,k,0,0) =& \cost(k, k+1, 0, 0) =  0, \\
\cost(k, \ell+1, 0, 0) =& \cost(k, \ell, 0, 0) + \del(b_k, b_{\ell+1}, b_\ell) , \\  
\end{aligned}
\end{equation*}
 \begin{equation}
 \begin{aligned}
 \cost(k, \ell+1, i, i) =&\min  \left \{
      	\begin{array}{ll}
  \cost(k, \ell, i, i) + \del(a_i, b_{\ell+1}, b_\ell), \\
\cost(k+1, \ell+1, i, i)+\del(b_k, a_i, b_{k+1}), \\
\cost(k, \ell+1, 0, 0) + \ins(b_k, b_{\ell+1}, a_i).
	\end{array}
	\right. 
 \end{aligned}
 \end{equation}
  \begin{equation}
 \begin{aligned}
 \label{Eqn:complexity}
 \cost(k, \ell, i, j+1) =&\min_{\substack{
 k',\ell' \in [k,\ell] \\
 i' \in[i,j] \\
 j'\in [i+1, j+1]
 } }   \left \{
      	\begin{array}{ll}
\cost(k, \ell', i, j+1)  + \del(a_{j+1}, b_\ell, \{b_{\ell'}, \cdots, b_{\ell-1}\}),\\
\cost(k', \ell, i, j+1) + \del(b_k, a_{i}, \{b_{k+1}, \cdots, b_{k'}\}), \\
\cost(k, \ell',i, i')+\cost(\ell', \ell,  j', j+1) + \del(a_{i'}, a_{j'}, b_{\ell'})\\ 
 + \ins(a_{i'}, a_{j'}, \{a_{i'+1}, \cdots, a_{j'-1}\})
	\end{array}
	\right. 
 \end{aligned}
 \end{equation}
 
   
\begin{lemma}
$\cost(k,l,i,j)$ provides the minimum cost converting from a trajectory substring $b_k b_{k+1} \cdots b_{\ell-1} b_\ell$ to $b_k  a_i a_{i+1} \cdots a_j b_\ell$ with arbitrary order of insertions and deletions. 
\end{lemma}
\begin{proof}
First, the claim is true for $\cost(k,\ell, 0,0)$ and $\cost(k,\ell, i, i)$. 
$\cost(k,\ell, 0, 0)$ is minimum cost of deleting all characters in $B$ between $b_k$ and $b_\ell$, which is the same in arbitrary order. 
Now we prove $\cost(k,\ell+1,i,i)$ is the minimum cost converting $b_kb_{k+1}\cdots b_{\ell+1}$ to $b_k a_i b_{\ell+1}$. Assume all $\cost(k', \ell', i,i)$ are computed. 
Now there are only three cases to consider according to the last operation we take. The first case is that we convert $b_k b_{k+1}\cdots b_{\ell}$ to $b_k a_i b_{\ell}$ and then delete $b_\ell$. 
The second case is that we convert $b_{k+1} b_{k+2}\cdots b_{\ell+1}$ to $b_{k+1} a_i b_{\ell+1}$ and then delete $b_{k+1}$. 
The last case is that we delete all characters between $b_k$ and $b_{\ell+1}$ and then insert $a_i$. The minimum cost of all three cases is $\cost(k,\ell+1,i,i)$, which provides the minimum cost converting $b_kb_{k+1}\cdots b_{\ell+1}$ to $b_k a_i b_{\ell+1}$. 

Now we prove $\cost(k,\ell,i,j+1)$ provides the minimum cost converting from a trajectory substring $b_k b_{k+1} \cdots b_{\ell-1} b_\ell$ to $b_k  a_i a_{i+1} \cdots a_j b_\ell$ with arbitrary order of insertions and deletions. 
Assume all $\cost(k,k',i',j'),\cost(\ell',\ell, i',j')$ are computed for  
 $k',\ell' \in [k,\ell],  i' \in [i,j],  j' \in [i+1, j+1]$.
We also have three cases with respect to the last operation.  We treat deletions or insertions of consecutive characters in a string as one operation. As for deletion, there are only two cases that are not handled by previous computation.  
The first case is that we convert $b_{k} b_{k+1}\cdots b_{\ell'}$ to $b_{k} a_i a_{i+1} \cdots a_{j+1} b_{\ell'}$ and then delete $b_{\ell'} b_{\ell'+1}\cdots b_{\ell-1}$. The second case is that we convert $b_{k'} b_{k'+1}\cdots b_{\ell}$ to $b_{k'} a_i a_{i+1} \cdots a_{j+1} b_{\ell}$ and then delete $b_{k+1} b_{k+2} \cdots b_{k'}$. There is only one case not computed if the last operation is an insertion. It is described in the following. 
We convert $b_k b_{k+1} \cdots b_{\ell'}$ to $b_k a_i a_{i+1} \cdots a_{i'} b_{\ell'}$ and convert $b_{\ell'} b_{\ell'+1} \cdots b_{\ell}$ to $b_{\ell'} a_{j'} a_{j'+1} \cdots a_{j+1} b_{\ell}$. Then 
we delete $b_{\ell'}$ and finally insert $a_{i'+1} a_{i'+2} \cdots a_{j'-1}$. The minimum cost of all these three cases is  $\cost(k,\ell,i,j+1)$, which provides the minimum cost converting from  $b_k b_{k+1} \cdots b_{\ell-1} b_\ell$ to $b_k  a_i a_{i+1} \cdots a_j b_\ell$. 
\end{proof}
With the recursion above, $\cost(0,m+1,0,n+1)$ is the metric-based edit distance between trajectory $A$ and $B$. The complexity is decided by the third case in Equation \ref{Eqn:complexity}, which is $O(n^3 m^3 (n+m))$. Therefore, we have the following lemma. 
\begin{lemma}
The algorithm computes the metric-based edit distance in $O(n^3m^3(n+m))$. 
\end{lemma}

%% file: insertFirstEdit.tex
\subsection{Algorithm for Insertion-first Metric-Based Edit Distance}
\label{sec:insertFirst}
\begin{wrapfigure}{r}{0.5\textwidth}
\centering
\vspace*{-4mm}
\includegraphics[width=0.3\textwidth]{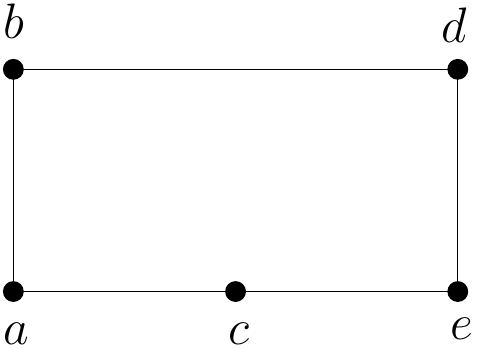}
\caption{In this figure, the coordinates of location characters are as following, $a=(0,0),b=(0,1), c=(1,0),d=(2,1), e=(2,0)$. We have trajectories with string representations, $A=ae,B=ace, C=acde$. $D_i(B,A)=0$, $D_i(B,C)=4\sqrt{2}-2$, $D_i(A,C)=2$. We have $D_i(B,C)>D_i(B,A)+D_i(A,C)$. The triangle inequality is not satisfied. }
\label{Fig:nonmetric}
\end{wrapfigure}

We also propose an insertion-first metric-based edit distance. That is, we require all insertions are done before any deletion, when converting from one trajectory string to the other.  This insertion-first version does not satisfy the triangle inequality. Denote this distance as $D_i(\cdot, \cdot)$. 
An example can be found in Figure \ref{Fig:nonmetric}.

Given two trajectories with string representation  $A=a_0a_1\dots a_na_{n+1}$ and $B=b_0b_1\dots b_mb_{m+1}$, where $a_0=b_0=S, a_{n+1}=b_{m+1}=T$, we first insert all characters in $A$ to $B$ obtaining an intermediate string, and thereafter, we delete all characters in $B$ from the intermediate string. 
Now we calculate the minimum cost converting $B$ to $A$ with all insertions before any deletion. 
Let $\DP(i,j,\alpha, \beta)$  denote the minimum cost of converting $b_0 \dots b_j$ to $a_0 \dots a_i$ on the conditions defined by $\alpha$ and $\beta$. 
$\alpha,\beta$ are used to guarantee that  the cost of insertions and deletions are calculated with respect to the neighbors. 
 If $\alpha=0$, the last operation of the conversion is inserting $a_i$, while  if $\alpha=1$, the last operation is deleting $b_j$.
 $\beta=0$ means after converting $b_0 \dots b_j$ to $a_0 \dots a_i$  the next character in the intermediate string is $a_{i+1}$, while $\beta=1$ means the next character is $b_{j+1}$.
$\DP(i,j,\alpha, \beta)$ is defined by recurrence in Equation \ref{DP}.

\newcounter{mytempeqncnt}
\begin{figure*}[!ht]
\normalsize

\setcounter{mytempeqncnt}{\value{equation}}

 \begin{equation}
 \label{DP}
   \begin{aligned}
     \DP(0,0,\alpha,\beta) &=  0 \quad  \alpha\in\{0,1\}, \beta\in \{0,1\} \\
     \text{For  } & 0<i\le n, 0<j\le m, & \\
    \DP(i,0,\alpha,\beta) &= \left \{ 
      	\begin{array}{ll}
	  \DP(i-1,0,0,0)+\ins(a_{i-1}, b_1, a_i)    & \alpha=0,\beta\in\{0,1\}\\
	  \infty & \alpha=1, \beta\in \{0,1\} \\
	\end{array}
	\right.
	\\
     \DP(0,j,\alpha,\beta)&=\left \{
      	\begin{array}{ll}
	  \DP(0,j-1,1,1) + \del(a_0,b_{j+1}, b_j) & \alpha=1,\beta=1 \\
	  \DP(0,j-1,1,1)+\del(a_0, a_1, b_j)    & \alpha=1,\beta=0 \\
	  \infty & \alpha=0, \beta\in\{0,1\} \\
	\end{array}
	\right.
	\\
	 \DP(i,j,\alpha,\beta) &=
     \left \{
      	\begin{array}{ll}
	  \min\{\DP(i-1,j,0,0) + \ins(a_{i-1},b_{j+1},a_i), \\
      \qquad \DP(i-1,j,1,0)+\ins(b_j, b_{j+1},a_i)\} & \alpha=0,\beta\in\{0,1\} \\
	  \min\{\DP(i,j-1,0,1) ,\DP(i,j-1,1,1)\}+ \del(a_{i},a_{i+1}, b_j) & \alpha=1,\beta=0 \\
	  \min\{\DP(i,j-1,0,1), \DP(i,j-1,1,1)\}+\del(a_i, b_{j+1},b_j) & \alpha=1,\beta=1 \\
	\end{array}
	\right.
   \end{aligned}
 \end{equation}

\setcounter{equation}{\value{mytempeqncnt}}
\end{figure*}




 With the recurrence, we can run dynamic programming to calculate $\DP(i,j,\alpha,\beta)$ with  time complexity of $O(mn)$.
 The insertion-first metric-based edit distance is $\min \{\DP(n+1, m+1, \alpha, \beta) | \alpha,\beta \in\{0,1\}\}$.
 \begin{lemma}
   $\DP(i,j,\alpha,\beta)$ provides the minimum cost of transforming from $b_0b_1\cdots b_i$ to $a_0a_1 \cdots  a_j$.
 \end{lemma}
 \begin{proof}
First, it is initially true on row $0$ and column $0$. 

Assume it is true for  $\DP(i-1,j,\alpha,\beta)$, $\DP(i,j-1,\alpha,\beta)$. We calculate $\DP(i,j,\alpha,\beta)$ transforming $b_1,b_2,\cdots,b_j$ to  $a_1,a_2,\cdots,a_i$. 
If $\alpha=0$,  let's consider  $\beta=0$, where the last operation is inserting $a_i$ and the next character in the intermediate string is $a_{i+1}$. If the previous character of $a_i$ in the intermediate string is $b_j$,
the insertion cost is $\ins(b_j,a_{i+1},a_i$, while if the previous is $a_{i-1}$, 
the insertion cost is $\ins(a_{i-1},a_{i+1},a_i)$. The minimum of $\DP(i-1,j,0,0) + \ins(a_{i-1},a_{i+1},a_i)$ and $\DP(i-1,j,1,0)+\ins(b_j,a_{i+1}, a_i)$
provides the minimum cost. The same goes with $\alpha=0,\beta=1$.

If $\alpha=1$, let's consider $\beta=0$, where the last operation is deleting $b_j$ and the next character in the intermediate string is $a_{i+1}$. 
The previous character of $b_j$ when deleting $b_j$ must be $a_{i-1}$, since all the $b_1,\cdots,b_{j-1}$ are deleted, and $a_{i}$ hasn't been inserted.
The minimum of $ \DP(i,j-1,0,1) ,\DP(i,j-1,1,1)$ plus $ \del(a_{i},a_{i+1}, b_j)$ is the minimum cost.
The same goes with $\alpha=1, \beta=1$.

All above, $\DP(i,j,\alpha,\beta)$ provides the minimum cost of transforming from $b_1,b_2,\cdots,b_i$ to $a_1,a_2,\cdots, a_j$.  
 \end{proof}

 \begin{lemma}
The algorithm computes the metric-based edit distance in running time $O(nm)$. 
 \end{lemma}



%% file: appendixForKGather.tex
\subsection{Proof for Theorem \ref{thm:hardedit}}
\label{sec:hardedit}
\begin{proof}
\label{proof:hardedit}
The proof for edit distance and metric-based edit distance is similar. 
	We first show that $k$-gather of trajectories on edit distance is NP-hard by reduction from the $3$SAT problem, where each variable appears at most $3$ times and each literal belongs to at most $2$ clauses \cite{garey2002computers}. 
    We have a boolean formula $\mathcal{F}$ in $3$-CNF form with $m$ clauses $C_j$ and $n$ variables $x_i$. $\mathcal{F} = C_1 \wedge \dots \wedge C_m$. 
    Each clause contains $3$ literals connected by $\vee$, where a literal is either a positive literal $x_i$ or a negative literal $\overbar{x_i}$, which is the negation of $x_i$.  
	From the boolean formula, we create a trajectory set, such that there is a solution to the $k$-gather problem with the radius of the edit distance of $5$, if and only if $\mathcal{F}$ has a satisfying assignment.

	We construct the gadget as following. First we define a planar graph on which the trajectories visit. Each face of the planar graph represents a character in the string representation of trajectories. The faces are clustered into $4n$ cells $c_1, c_2,\cdots, c_{4n}$. The first $3n$ cells each consists of $3$ faces each: $2$ rectangle faces side by side and $1$ triangle face on top of them, while for the last $n$ cells each consists of $2$ rectangle faces side by side and $2$ triangle faces adjacent to both rectangle faces as shown in Figure \ref{fig:NPproof}. 
	
All trajectories go from $c_1$ to $c_{4n}$ through all the rectangular faces. But some trajectories take `detours' to visit some of triangular faces. We call that a top detour (visiting the face above the rectangular faces) or a bottom detour (visiting the face below the rectangular faces). 
For the part of trajectories within each cell, the edit distance between a trajectory with a detour and a trajectory going straight is $1$, and the edit distance between a trajectory with a top detour and  one with a bottom detour is $2$.
	
Now we  introduce variable gadgets. For each variable $x_i$, we construct two trajectories for variable assignment $x_i=\mbox{True}$ and $x_i=\mbox{False}$ respectively. Both trajectories have top detours in $c_{3i-2}, c_{3i-1}, c_{3i}$; the trajectory of $x_i=\mbox{True}$ has a top detour in cell $c_{ 3n+i}$ while the trajectory of $x_i=\mbox{False}$ has a bottom detour in that cell; and there are no detours for all the other cells in the two trajectories. The edit distance between the trajectory for $x_i=\mbox{True}$ and $x_i=\mbox{False}$ is $2$.  
	In addition, we construct $k-3$ more supplement trajectories for $x_i$ which have top detours only in $c_{3i-2}, c_{3i-1}, c_{3i}$ and no other detours in the rest cells.  The edit distance between these supplement trajectories and $x_i=\mbox{True},\mbox{False}$ is $1$.  
	For each clause, we create $3$ trajectories such that for each variable $x_i$ in the clause, each of the three trajectories has a top detour in cell $c_{3n+i}$ , else if $\overbar{x_i}$ is in the clause, each of the three trajectories has a bottom detour in cell $c_{3n+i}$; the three trajectories have no detours in the rest cells. 
	
	If  $\mathcal{F}$ is satisfiable, there is a solution for $k$-gather with respect to the assignment. For each variable $x_i$, if $x_i=\mbox{True}$, we pick the trajectory of $x_i=\mbox{True}$ as the center of a cluster; otherwise we pick the trajectory of $x_i=\mbox{False}$ as the center of a cluster. To form the cluster, we first add the trajectory of the other variable assignment to the cluster. 
    Thereafter,  we add the $k-3$ supplement trajectories with top detours in cell $c_{3i-2},c_{3i-1},c_{3i}$. 
    Finally, add a trajectory of a clause containing literal $x_i$ if $x_i=\mbox{True}$, or literal $\overbar{x_i}$ if $x_i=\mbox{False}$,  to this cluster. It is guaranteed to get such a trajectory of a clause since we have three copies for each clause with three variables.  
Thus, we form a cluster with $k$ trajectories. We can obtain a $k$-gather clustering by repeating the above process on all variables and if any trajectory of a clause is left after all clusters are formed, we can simply add the clause trajectory to any cluster centered on the literals it contains. 
	The maximum radius of all such clusters is $5$, which is the distance between a trajectory of $x_i=\mbox{True/False}$ and a trajectory of a clause containing the corresponding literal $x_i/\overbar{x_i}$. 

If we have a $k$-gather clustering with radius $5$, now we prove that we can obtain a satisfying assignment of $\mathcal{F}$ from the clustering. First, trajectories of $x_i=\mbox{True}$ and $x_i=\mbox{False}$ have at most $k+6$ other trajectories within distance of $5$.  Since each variable cannot appear more than three times in the clause, there are at most $9$ trajectories of clauses within distance $5$ from these two trajectories.  Adding to the $k-3$ supplement trajectories for $x_i$ within distance $1$ 
   which have detours only on cell $c_{3i-2},c_{3i-1},c_{3i}$, we get $k+6$ other trajectories within distance $5$ from trajectories of $x_i=\mbox{True}$ and $x_i=\mbox{False}$. Hence, 
there are not enough trajectories to make both of them as centers of two clusters within radius $5$, if $k>13$. 
	At most one of the trajectories of $x_i=\mbox{True}$ and $x_i = \mbox{False}$ can be chosen as the center of a cluster.  
Second, as	for a trajectory of a clause, since each literal can appear at most twice in clauses, there are at most $9$ other trajectories representing clauses within distance $5$. 
Each clause has at most $3$ literals, therefore, there are at most $3$ trajectories of variables within distance $5$ from a trajectory representing a clause. 
In all, there are at most $12$ other trajectories within distance $5$ from a trajectory of a clause. 
For $k>13$, a trajectory of a clause can not be the center of a cluster. 
Third, a supplement trajectory for $x_i$ with top detours only on cell $c_{3i-2},c_{3i-1},c_{3i}$, has $k-2$ other trajectories within distance $5$, including the $k-3$ same trajectories and $2$ trajectories of $x_i=\mbox{True}$ and $x_i=\mbox{False}$. 
Therefore, a supplement trajectory cannot be the center of a cluster. 
All above, exactly one of the trajectories of $x_i=\mbox{True}$ and $x_i=\mbox{False}$ can be chosen as the center of a cluster.  
A clause is satisfied by the variable assignment of the center trajectory, when the clause is contained in a cluster. 
Therefore, those trajectories acting as centers of clusters provide us the assignment of $\mathcal{F}$.
	
	All above, there is a solution to the $k$-gather problem with the maximum radius of $5$  if and only if $\mathcal{F}$ has a satisfying assignment. 
    We can use the same proof for metric-based edit distance, 
    if we set the distance from the center of a triangle face to the center of a rectangle face in each cell to be $1$.
\end{proof}

\begin{figure*}[htpb]
	\centering
	\includegraphics[width=\textwidth]{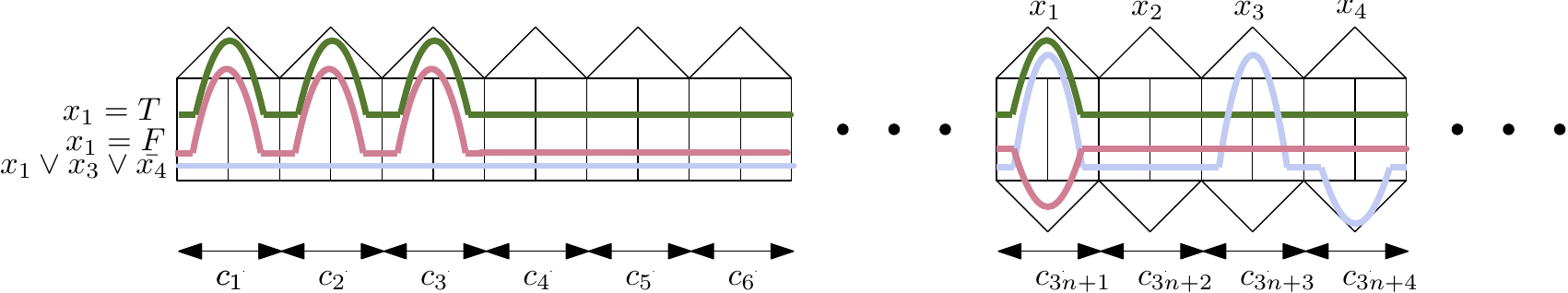}
	\caption{The example of the gadget.}
	\label{fig:NPproof}
\end{figure*}

\subsection{Hardness on the Jaccard Distance}
\label{sec:hardnessJacc}
A trajectory with string representation can be processed as a set of \emph{$w$-shingles}. 
\begin{definition}[$w$-shingle]
A $w$-shingle is a sequence of $w$ consecutive location strings on trajectory with string representation.
\end{definition}
We combine all $w$-shingles from each trajectory into a set. 
Thereafter, we introduce Jaccard distance over the set of $w$-shingles as a metric between two trajectories. 
\begin{definition}[Jaccard distance]
For two sets $A$ and $B$,
the \emph{Jaccard distance} is defined as $$d(A,B) = 1-\frac{\left | A \cap B \right |}{\left | A \cup B \right |}.$$

\end{definition}

Computing the Jaccard distance between two trajectories with $O(n)$ locations  can be done in $O(nw)$ time. 

\begin{theorem}
	The $k$-gather problem on trajectories with Jaccard distance is NP-hard. 
\end{theorem}
\begin{proof}
	The proof use the same gadget as Theorem \ref{thm:hardedit}. 
There is a solution to the $k$-gather problem if and only if the maximum radius of clusters with Jaccard distance on $2-$shingles is $\frac{15}{8n+10}$. 

\end{proof}

\subsection{Approximation Algorithm}
\label{sec:approxAlgo}

Fortunately for the $k$-gather problem there is a $2$-approximation algorithm as proved in \cite{aggarwal2006achieving}. 
We provide the $2$-approximation algorithm similar to \cite{aggarwal2006achieving} to find out the minimum radius for $k$-gather problem with edit distance, metric-based edit distance and Jaccard distance to make our work complete.
First, we compute the distance for all pairs of trajectory strings.
We try all the values of distance to be the radius $R$ of the $k$-gather clustering, so that  we find the smallest $R$ satisfying the two conditions:
\begin{enumerate}
  \item  There are at least $k-1$ other trajectories within  the distance of $R$ taking each trajectory as the center. 
  \item Initially mark all the trajectories as uncovered , let the set of center of cluster $C$ be an empty set. Denote the set of trajectories  as $V$. 
    Repeat the following procedure until all trajectories are covered: select an arbitrary unmarked trajectory $v\in V$ as a center. If there is any uncovered trajectory within the distance of $R$ from $v$,  we add $v$ to $C$, and mark all the  trajectories within distance $R$ from $v$ as covered.
   
    Thereafter, we create a flow network as following. Create a source $s$,  sink $t$, nodes for each center in $C$ and nodes for each trajectory in $V$. For each center $c\in C$, add edges from $s$ to $c$  with capacity $k$; meanwhile add edges with unit capacity from $c$ to all the trajectories that are within distance $R$ from $c$ .
    For each trajectory $v \in V$, add edges from $v$ to $t$ with unit capacity.
    Check if there is a flow of capacity $k\left | C \right |$ from $s$ to $t$. 
    If so, the flow provides a solution for $k$-gather. If not, exit with failure.

\end{enumerate}

%% file: appendixForDF.tex
\subsection{Proof for Lemma \ref{lem:parralel}}

\begin{proof}
Any coupling $\beta$ starts from the pair $(A_{i,1}, B_{j, 1})$ and ends at $(A_{i,D}, B_{j, D})$. 
If $\beta$ is not a parallel coupling, there must exist a pair $(A_{i,k}, B_{j,k'})$ in $\beta$ such that $k$ and $k'$ are one odd and one even. 
The corresponding points are on different sides of the horizontal axis. 
It forces $\delta(\beta) \geq \min \{d(a_1^o, b_0^e), d(a_0^o, b_1^e) \} = 1.65>1.61$. 

Now we consider the case that $\beta$ is a parallel coupling. 
If $u_i\perp v_j$, for each pair of points in the coupling, at least one point is in $\{a_0^o, a_0^e, b_0^o, b_0^e \}$. Since $d(a_0^p, b_{\{0,1\}}^p)=1, d(b_0^p, a_{\{0,1\}}^p)=1$, for $p\in \{o,e\}$, $\delta(\beta) \leq 1$. 
If $u_i \not \perp v_j$, there  exists at least one pair of points in the coupling that are either $(a_1^o, b_1^o)$ or $(a_1^e, b_1^e)$. Therefore, $\delta(\beta) \geq \min\{d(a_1^o, b_1^o), d(a_1^e, b_1^e) \} = 1.61$.
\end{proof}

\subsection{Proof for Lemma \ref{lem:dF1}}
\begin{proof}
\label{proof:dF1}
To be simple, for a pair $(p,q)$ in a coupling of $P,Q$, we just denote that $P$ goes to $p$ and $Q$ goes to $q$. 
For $u_i \perp v_j$, the coupling of  $P$ and $Q$ is constructed  as follows: 
\begin{enumerate}
\item $P$ goes through $D(N-j)$ points of $W$ while $Q$ stays at $\omega_1$. 
\item We do a parallel coupling between $B_k$ for $k=1,\dots,j-1$ and the rest of $W$. When $Q$ goes to $\omega_1$ and $\omega_2$ at the start and end of each $B_k$, $P$ stays at $a_0^e$. 
\item $P$ goes to $x_1$ and $Q$ goes to $\omega_1$. $Q$ stays at $\omega_1$ until $P$ proceeds and goes to the $s$ before $A_i$. 
\item $P$ and $Q$ traverse $A_i$ and $B_j$ in parallel. $P$ goes to $s$ and  $Q$ goes to $\omega_2$.  
\item $P$ continues to $x_2$ while $Q$ stays at $\omega_2$. 
\item $Q$ goes to $\omega_1$ and does the same parallel coupling as step 2 for the rest of $B_k$. 
\item $P$ finishes the rest of $W$, while $Q$ stays at $\omega_2$. 
\end{enumerate}

 In step 1, the  distances are $d(a_0^{\{o,e\}, \omega_1})=0.67$. 
In step 2, $d(a_0^e, \omega_1)=0.67, d(a_0^e, \omega_2)=1$, all the other distances in the parallel coupling are $1$. 
In step 3,  $d(\omega_1,s)=d( \omega_1, a_1^{\{o,e\}}) = 1$. 
In step 4, the distances in the parallel coupling is $1$ by Lemma \ref{lem:parralel}, $d(\omega_2, s) = 0.11$. 
In step 5, the distances between all points in $A_k$ and $\omega_2$ is at most $1$, $d(\omega_2, x_2) = 1$. 
In step 6 and 7, all the distances are no greater than $1$. 
In each step, the distance between each pair of points in $P$ and $Q$ is no greater than $1$. 

Therefore, $\delta_{dF}(P, Q)\leq 1$. 
\end{proof}

\subsection{Proof for Lemma \ref{lem:perp}}
\begin{proof}
\label{proof:perp}
In a coupling $\beta$, when $P$ goes to $x_1$, if $Q$ goes to $b_0^e, b_1^e, \omega_2$, $\delta(\beta)\geq 1.61$, since $d(x_1,b_0^e)=2.66, d(x_1,b_1^e)=1.61, d(x_1, \omega_2)=1.61 $. Then we are done.
If $Q$ goes to $b_0^o, b_1^o$, $d(x_1, b_0^o) =0.49, d(x_1, b_1^o) =1.13$. In the next step, either $P$ goes to $s$ or $Q$ goes to $b_1^e,b_0^e$ (or both). 
Since $d(s, b_0^{\{o,e\}}) = 1.70, d(s, b_1^{\{o,e\}})=1.97$, it also implies $\delta(\beta)>1.61$. 

Now we analyse the case that $P$ goes to $x_1$, and $Q$ goes to $\omega_1$. Consider the case when $Q$ leaves $\omega_1$ and goes to $b_{\{0,1\}}^o$. $P$ must be at $a_0^o, a_1^o$, otherwise, we would get $\delta(\beta)>1.61$. If the point $P$ goes to belongs to the sequence of $W$, it means $P$ has already passed $x_2$ when $Q$ stays at $\omega_1$. $d(x_2, \omega_1)=1.89$. $\delta(\beta)>1.61$. 
Now consider the case that the point $P$ goes to belongs to $A_i$. $P$ and $Q$ must go simultaneously in the same side below or above the horizontal axis. Otherwise, $\delta(\beta)>1.61$. If the coupling is not a parallel coupling between some $A_i$ and $B_j$, then $P$ must first arrive at $s$. The distance between $s$ and any point in $B_j$ is greater than $1.61$. $\delta(\beta)>1.61$. If the coupling is a parallel coupling, we still have $\delta(\beta)>1.61$
by Lemma \ref{lem:parralel}, since no $u_i\perp v_j$. 

All above, we have $\delta_{dF}(P,Q)\geq 1.61$. 
\end{proof}